\documentclass[12pt,a4paper,fullpage,psfig]{article}

\usepackage{color}
\usepackage{amsmath}
\usepackage{amssymb}
\usepackage{latexsym}
\usepackage{epsfig,epsf,psfrag,geometry,bm}
\usepackage{amsthm}

\usepackage{natbib} 
\bibpunct{(}{)}{;}{a}{,}{,} 

\renewcommand{\phi}{\varphi}
\newcommand{\hist}{{\cal H}}

\newcommand{\dd}{\textup{d}}
\newcommand{\E}{\mathbb{E}}
\newcommand{\p}{\mathbb{P}}
\newcommand{\Exp}{\textup{Exp}}

\newtheorem{proposition}{\noindent Proposition}[section]

\theoremstyle{definition}
\newtheorem{example}{\noindent Example}[section]

\newtheorem{algorithm}{\noindent Algorithm}[section]

\title{Lecture Notes: \\Temporal Point Processes\\ and the Conditional Intensity Function}
\author{Jakob Gulddahl Rasmussen\\ Department of Mathematical Sciences\\ Aalborg University\\ Denmark\\ jgr@math.aau.dk}

\begin{document}

\maketitle

\abstract{These short lecture notes contain a not too technical introduction to point processes on the time line. The focus lies on defining these processes using the conditional intensity function. Furthermore, likelihood inference, methods of simulation and residual analysis for temporal point processes specified by a conditional intensity function are considered.}

\newpage
\tableofcontents
\newpage

\section{Introduction}

A temporal point pattern is basically a list of times of events.  Many
real phenomena produce data that can be represented as a temporal
point pattern; the left column of Table~\ref{tab.examples} shows a few
examples. Common to these examples is that we do not know how many
events will occur, or at what times they will occur. Usually complex
mechanisms are behind these seemingly random times, for example
earthquakes cause new earthquakes in the form of aftershocks. An
essential tool for dealing with these mechanisms, for example in
predicting future events, is a stochastic process modelling the point
patterns: a {\em temporal point process}. The term point is used since
we may think of an event as being instant and thus can represent it as
a point on the time line. For the same reason the words point and
event will be used interchangeably throughout this note.

\begin{table}
  \begin{center}
    \begin{tabular}{|l|l|}
      \hline
      Events & Marks\\
      \hline
      Earthquakes & Magnitudes \\
      & Locations \\
      \hline
      Arrivals at a server & Service time\\
      \hline
      Accidents & Insurance claims\\
      & Type of Injury\\
      \hline
    \end{tabular}
  \end{center}
  \label{tab.examples}
  \caption{Examples of events and marks.}
\end{table}

Often there is more information available associated with an
event. This information is known as marks. Examples are given in the
right column of Table~\ref{tab.examples}. The marks may be of separate
interest or may simply be included to make a more realistic model of
the event times. For example, it is of practical relevance to know
the position and magnitude of an earthquake, not just its time. At the
same time, the magnitude of an earthquake also influences how many
aftershocks there will be, so a model not including magnitudes as
marks may not be reliable at modelling the event times either.

In this note, familiarity with the Poisson process on the line as well
as basic probability theory and statistics is assumed. On the other
hand, measure theory is not assumed; for a much more thorough
treatment with all the measure theoretical details, see
\cite{daley-vere-jones-03} and \cite{daley-vere-jones-08}.

\section{Evolutionary point processes}

There are many ways of treating (marked) temporal point processes. In
this note we will explore one approach based on the so-called
conditional intensity function. To understand what this is, we first
have to understand the concept of evolutionarity.

\subsection{Evolutionarity}

Usually we think of time as having an {\em evolutionary character}:
what happens now may depend on what happened in the past, but not on
what is going to happen in the future. This order of time is also a
natural starting point for defining practically useful temporal point
processes. Roughly speaking, we can define a point process by
specifying a stochastic model for the time of the next event given we
know all the times of previous events. The term {\em evolutionary
  point process} is used for processes defined in this way.

The past in a point process is captured by the concept of the {\em
  history} of the process. If we consider the time $t$, then the
history $\hist_{t-}$ is the knowledge of times of all events, say
$(\ldots,t_1,t_2,\ldots,t_n)$, up to but not including time $t$;
$\hist_t$ also includes the information whether there is an event at
time $t$. Note that theoretically the point process may extend
infinitely far back in time, but it does not have to do this. Note
also that we assume that we have a {\em simple point process}, i.e.\ a
point process where no points coincide, such that the points can be
strictly ordered in time.


\subsection{Interevent times}

When specifying a temporal point process we can use many different
approaches. In this note, we start by specifying the
distribution of the time lengths between subsequent events, and then in the next section we reformulate this in terms of conditional intensity functions.

The lengths of the time intervals between subsequent events are known
as {\em interevent times}. We can define a temporal point process by
specifying the distributions of these. Let $f(t_{n+1}|\hist_{t_n})$ be
the conditional density function of the time of the next event
$t_{n+1}$ given the history of previous events $(\ldots,t_{n-1},t_n)$.
Note that the density functions $f(t_n|\ldots,t_{n-2},t_{n-1})$
specify the distributions of all interevent times,
one by one, starting in the past, and thus the distribution of all
events is given by the joint density
\[
f(\ldots,t_1,t_2,\ldots) = \prod_n f(t_n|\ldots,t_{n-2},t_{n-1}) =
\prod_n f(t_n|\hist_{t_{n-1}})
\] 
in the same manner as the joint density for a bivariate random
variable factorises into $p(x,y) = p(x) p(y|x)$. Let us consider a
simple example of a point process defined by specifying the density
function for interevent times:

\begin{example}[Renewal process and Wold process]\label{ex.ren}
  The simplest process we can define by specifying the distribution of
  the interevent times is the renewal process. This process is defined
  by letting the interevent times be i.i.d.\ stochastic variables,
  i.e.\ $f(t_n|\hist_{t_{n-1}})=g(t_n-t_{n-1})$ where $g$ is a density
  function for a distribution on $(0,\infty)$. An important special case of this is the homogeneous Poisson process with intensity $\lambda$, where $g$ is the density of the exponential distribution with inverse mean $\lambda$. Figure~\ref{fig-renewal-processes} shows simulations of
  three different renewal processes: one is the homogeneous Poisson process, one is more {\em clustered} than the Poisson process (i.e.\ the points tend to occur in clusters), and one is more {\em regular} than the Poisson process (i.e.\ the points tend to be more evenly spread out).
  \begin{figure}[htbp]
    \begin{center}
      \leavevmode
      \includegraphics[height=3cm]{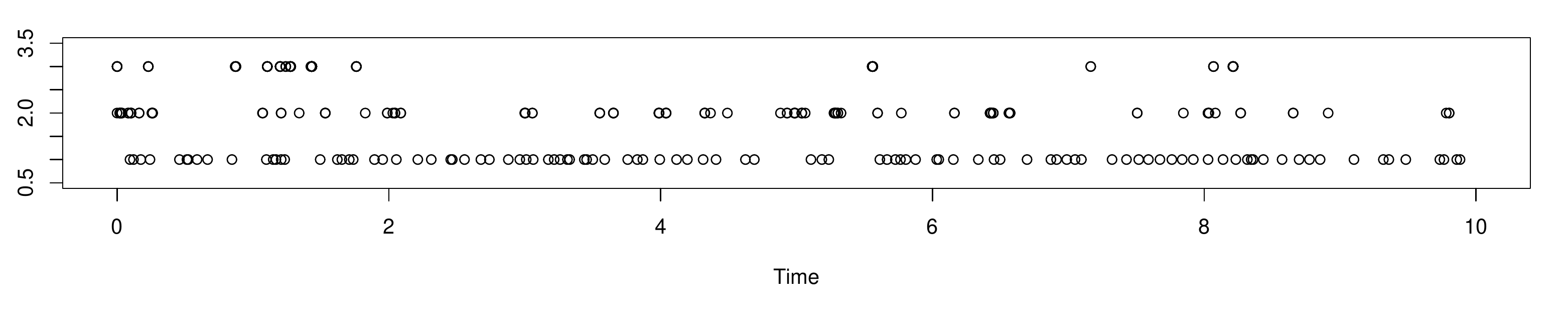}
      \caption{Three simulations of renewal processes with different
        interevent time distributions: Gamma(0.02,0.2) (upper),
        Gamma(0.1,1) (middle), Gamma(2,20) (lower). Note how the upper
        case is clustered and the lower case is regular compared to
        the middle case (which is a Poisson process). Also note that
        all the simulations have roughly 100 points for easy
        comparison (they are very densely packed together for the
        upper case).}
      \label{fig-renewal-processes}
    \end{center}
  \end{figure}
\end{example}

\subsection{Conditional intensity function}\label{sec.cif}

Example~\ref{ex.ren} show cases where $t_n$ depends only on
$t_{n-1}$. However, in general it may depend on the whole history, and it turns out that the density function of the interevent times is not the best way of specifying the general case. Instead the conditional intensity function is a more convenient and intuitive way
of specifying how the present depends on the past in an evolutionary
point process. Consider the conditional density $f(t|\hist_{t_n})$ and
its corresponding cumulative distribution function $F(t|\hist_{t_n})$
for any $t>t_n$. Then the {\em conditional intensity function} (or
hazard function) is defined by
\begin{equation}\label{eq.int}
  \lambda^*(t) = \frac{f(t|\hist_{t_n})}{1-F(t|\hist_{t_n})}.
\end{equation}
The conditional intensity function can be interpreted heuristically in
the following way: consider an infinitisemal interval around $t$, say
$\dd t$, then
\begin{eqnarray*}
\lambda^*(t)\dd t
&=& \frac{f(t|\hist_{t_n})\dd t}{1-F(t|\hist_{t_n})}\\
&=& \frac{\p(t_{n+1}\in[t,t+\dd t]|\hist_{t_n})}
    {\p(t_{n+1}\notin(t_n,t)|\hist_{t_n})}\\
&=& \frac{\p(t_{n+1}\in[t,t+\dd t],t_{n+1}\notin(t_n,t)|\hist_{t_n})}
    {\p(t_{n+1}\notin(t_n,t)|\hist_{t_n})}\\
&=& \p(t_{n+1}\in[t,t+\dd t]|t_{n+1}\notin(t_n,t),\hist_{t_n})\\
&=& \p(t_{n+1}\in[t,t+\dd t]|\hist_{t-})\\
&=& \E[N([t,t+\dd t])|\hist_{t-}],
\end{eqnarray*}
where $N(A)$ denotes the number of points falling in an interval $A$,
and the last equality follows from the assumption that no points
coincide, so that there is either zero or one point in an
infinitisemal interval. In other words, the conditional intensity
function specifies the mean number of events in a region conditional
on the past. Here we use the notation $*$ from
\cite{daley-vere-jones-03} to remind ourselves that this density is
conditional on the past right up to but not including the present,
rather than writing explicitly that the function depends on the
history.

We consider a few examples of point processes where the conditional
intensity has particular functional forms:

\begin{example}[Poisson process]
  The (inhomogeneous) Poisson process is among other things
  characterised by the number of points in disjoint sets being
  independent. The conditional intensity function inherets this
  independence. The Poisson process is quite simply the point process
  where the conditional intensity function is independent of the past,
  i.e. the conditional intensity function is equal to the intensity
  function of the Poisson process, $\lambda^*(t) = \lambda(t)$.
\end{example}

\begin{example}[Hawkes process]\label{ex.haw}
  Define a point process by the conditional intensity function
  \begin{equation}\label{eq.hawexp}
  \lambda^*(t) = \mu + \alpha\sum_{t_i<t}\exp(-(t-t_i)),
  \end{equation}
  where $\mu$ and $\alpha$ are positive parameters. Note that each
  time a new point arrives in this process, the conditional intensity
  grows by $\alpha$ and then decreases exponentially back towards
  $\mu$. In other words, a point increases the chance of getting other
  points immediately after, and thus this is model for clustered point
  patterns. A simulation of the process with parameters $(\mu,\alpha)
  = (0.5,0.9)$ is shown in Figure~\ref{fig-hawkes-process} together
  with its conditional intensity function (in Section~\ref{sec.sim} we
  will learn how to make such a simulation). The so-called Hawkes
  process is a generalization of this process and has the conditional
  intensity function
  \begin{equation*}
  \lambda^*(t) = \mu(t) + \alpha\sum_{t_i<t}\gamma(t-t_i;\beta),
  \end{equation*}
  where $\mu(t)\geq0$, $\alpha>0$, and $\gamma(t;\beta)$ is a density
  on $(0,\infty)$ depending on some parameter $\beta$ (which may be a single value or a vector, depending on the choice of distribution). For more on the
  Hawkes process, see
  e.g. \cite{hawkes-71a,hawkes-71b,hawkes-72,hawkes-oakes-74}.
  \begin{figure}[htbp]
    \begin{center}
      \leavevmode
      \includegraphics[height=6cm]{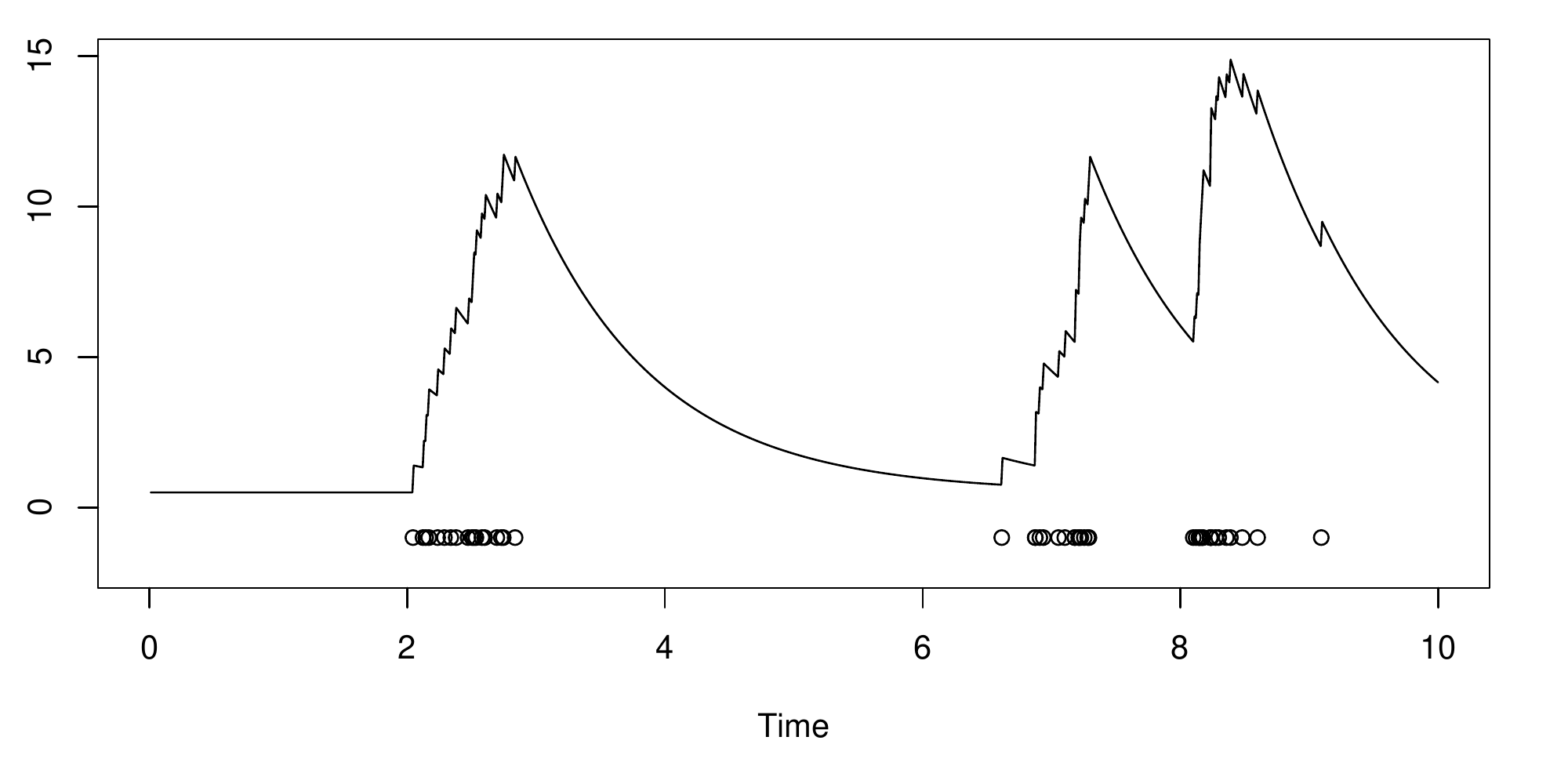}
      \caption{A simulation of the Hawkes process is shown at the
        bottom of this plot, and the corresponding conditional
        intensity function is shown in the top. Note that the point
        pattern is clustered.}
      \label{fig-hawkes-process}
    \end{center}
  \end{figure}
\end{example}

\begin{example}[Self-correcting process]\label{ex.inhib}
  What do we do if we want a point process for regular point patterns?
  Exchanging the plus for a minus in the Hawkes process will not work,
  since a conditional intensity function has to be non-negative. We can
  instead use
  \[
  \lambda^*(t) = \exp\left(\mu t - \sum_{t_i<t}\alpha\right),
  \]
  where $\mu$ and $\alpha$ are positive parameters. Now the intensity
  rises as time passes, but each time a new point appears we multiply
  by a constant $e^{-\alpha}<1$, and thus the chance of new points
  decreases immediately after a point has appeared; in other words,
  this is a regular point process. A simulated point pattern and the
  conditional intensity function is shown in
  Figure~\ref{fig-selfcorr}. This process is a special case of the
  so-called self-correcting process \citep{isham-westcott-79}.
  \begin{figure}[htbp]
    \begin{center}
      \leavevmode
      \includegraphics[height=6cm]{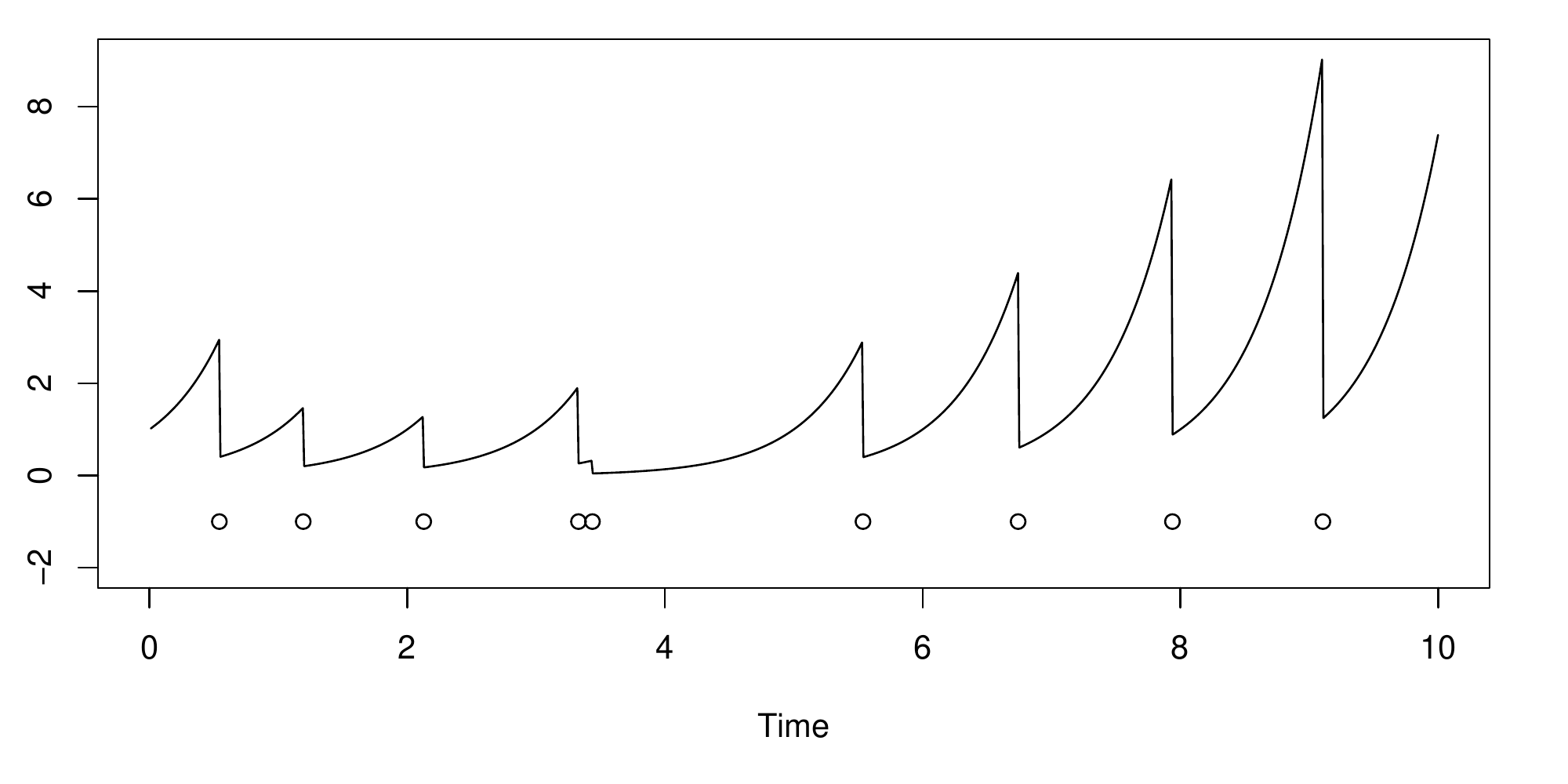}
      \caption{A simulation of a self-correcting process is shown at the
        bottom of this plot, and the corresponding conditional
        intensity function is shown in the top. Note that the point
        pattern is regular.}
      \label{fig-selfcorr}
    \end{center}
  \end{figure}
\end{example}

Note that the models in examples \ref{ex.haw} and \ref{ex.inhib} are
specified simply by choosing a particular form of the conditional
intensity and interpreting this. A little creativity and common sense
can be used to define many new models using the conditional intensity
function. This, of course, depends on the fact that the conditional
intensity function uniquely defines a point process. To prove this we
first need to note that the definition of the conditional intensity
function can also be reversed such that an expresion for the density
or cumulative distribution function of the interevent times can be
obtained:
\begin{proposition}\label{prop.fstar}
  The reverse relation of (\ref{eq.int}) is given by
  \begin{equation}\label{eq.fstar}
    f(t|\hist_{t_n})=\lambda^*(t)\exp\left(-\int_{t_{n}}^t\lambda^*(s)\dd s\right),
  \end{equation}
  or
  \begin{equation}\label{eq.Fstar}
    F(t|\hist_{t_n}) = 1-\exp\left(-\int_{t_{n}}^t\lambda^*(s) \dd s\right),
  \end{equation}
  where $t_n$ is the last point before $t$.
\end{proposition}

\begin{proof}
  By (\ref{eq.int}), we get that
  \begin{equation}\label{eq.l1F}
  \lambda^*(t) = \frac{f(t|\hist_{t_n})}{1-F(t|\hist_{t_n})} 
  = \frac{\frac{\dd}{\dd t}F(t|\hist_{t_n})}{1-F(t|\hist_{t_n})}
  = -\frac{\dd}{\dd t}\log(1-F(t|\hist_{t_n})).
  \end{equation}
  Integrating both sides, we get by the fundamental theorem of calculus
  that 
  \[
  \int_{t_n}^t\lambda^*(s) \dd s = -(\log(1-F(t|\hist_{t_n})) -
  \log(1-F(t_n|\hist_{t_n}))) = -\log(1-F(t|\hist_{t_n})),
  \]
  since $F(t_n|\hist_{t_n})=0$ (point $t_{n+1} = t_{n}$ with
  probability zero, since the point process is simple).  Isolating
  $F(t|\hist_{t_n})$ we get (\ref{eq.Fstar}), and (\ref{eq.fstar})
  then follows by differentiating $F(t|\hist_{t_n})$ with respect to
  $t$, again using the fundamental theorem of calculus.
\end{proof}

\begin{proposition}\label{prop.defuni}
  A conditional intensity function $\lambda^*(t)$ uniquely defines a
  point process if it satisfies the following conditions for any point pattern $(\ldots,t_1,\ldots,t_n)$ and any $t>t_n$:
  \begin{enumerate}
  \item $\lambda^*(t)$ is non-negative and integrable on any interval starting at $t_n$, and
  \item $\int_{t_{n}}^t\lambda^*(s)\dd
    s\rightarrow\infty$ for $t\rightarrow\infty$.
  \end{enumerate}
\end{proposition}

\begin{proof}
  The distribution of the point process is well-defined, if all
  interevent times have well-defined densities, i.e.\
  $f(t|\hist_{t_n})$ should be a density function on $t\in[t_n,\infty)$, or
  equivalently $F(t|\hist_{t_n})$ should be a cumulative distribution
  function. From the assumptions and (\ref{eq.Fstar}) it follows
  that
  \begin{itemize}
  \item $0 \leq F(t|\hist_{t_n}) \leq 1$,
  \item $F(t|\hist_{t_n})$ is a non-decreasing function of $t$,
  \item $F(t|\hist_{t_n})\rightarrow1$ for $t\rightarrow\infty$,
  \end{itemize}
  which means that $F(t|\hist_{t_n})$ is a distribution function. Uniqueness follows from Proposition~\ref{prop.fstar},
  since $F(t|\hist_{t_n})$ is uniquely obtained from $\lambda^*(t)$ using
  (\ref{eq.Fstar}).
\end{proof}

Note that item 2.\ in Proposition~\ref{prop.defuni} implies that the
point process continues forever, a property which is often not
desireable for practical use - luckily we can get rid of this
assumption. If we remove this, the proof still holds except that item
2.\ in the proof has to be removed. Now $F(t|\hist_{t_n})\rightarrow p$ for some
probability $p<1$, so we have to understand what it means when the
cumulative distribution function for the interevent time does not tend
to one when time tends to infinity. Basically this means that there is
only probability $p$ of having one (or more) points in the rest of the
process, and with probability $1-p$ the process terminates with no
more points.

\begin{example}[Two terminating point processes]
  Consider a unit-rate Poisson process on $[0,1]$. This
  has conditional intensity function
  $\lambda^*(t)={\bf1}[t\in[0,1]]$. Thus starting at zero (with no
  points so far), we get that
  \[
  F(t|\hist_0) = 1 - \exp\left(-\int_0^t{\bf1}[s\in[0,1]]\dd s\right)
  = 1 - \exp\left(-\min\{t,1\}\right),
  \]
  where ${\bf1}[\cdot]$ denotes the indicator function. For $t>1$,
  this equals $1-\exp(-1)\approx0.63$, so there is a probability of
  about $0.37$ of having no points at all. If we do get a point, say $t_1$, there
  is an even smaller chance of getting another point in the remaining interval $(t_1,1]$. Another terminating unit-rate process could be a process that
  behaves like a Poisson process but stops after $n$ points. In this
  case
  \[
  F(t|\hist_{t_i}) = (1 - \exp(-t)) {\bf1}[i<n].
  \]
  Both these examples illustrate that assumption 3. in
  Proposition~\ref{prop.defuni} is not necessary to get well-defined
  point processes.
\end{example}

\subsection{The marked case}

The conditional intensity function also generalises to the marked
case, but before we get that far it is worth reminding ourselves that
the mark space $\mathbb{M}$ can be many different types of spaces,
often (a subset of) $\mathbb{R}$ or $\mathbb{N}$. We can specify the
distribution of the mark $\kappa$ associated with the point $t$ by its
conditional density function $f^*(\kappa|t)=f(\kappa|t,\hist_{t-})$,
i.e.\ this specifies the distribution of the mark $\kappa$ given $t$
and the history $\hist_{t-}$, which now includes information of both
times and marks of past events. Here the term density function is used
in a broad sense: if the mark is a continuous random variable, this is
the usual (conditional) density function, but if it is a discrete
random variable, this is its (conditional) probability function. Note
also that $f^*(\kappa|t)=f(\kappa|t,\hist_{t_n})$ if $t_n$ is the the
last point before $t$, since the additional condition that the next
point is located at $t$ means that the histories
$\hist_{t-}$ and $\hist_{t_n}$ contain the same information.

We can now define the conditional intensity function for the marked
case as 
\[
\lambda^*(t,\kappa) = \lambda^*(t) f^*(\kappa|t),
\]
where $\lambda^*(t)$ is called the {\em ground intensity}, and is
defined exactly as the conditional intensity function for the unmarked
case, except that it is allowed to depend on the marks of the past
events also; note the close resemblance of this formula with
$p(x,y)=p(x)p(y|x)$ for the relation between the joint, marginal and
conditional distributions for random variables. Thus we can rewrite this expression to
\[
\lambda^*(t,\kappa) = \lambda^*(t) f^*(\kappa|t) =
\frac{f(t|\hist_{t_n})f^*(\kappa|t)}{1-F(t|\hist_{t_n})} =
\frac{f(t,\kappa|\hist_{t_n})}{1-F(t|\hist_{t_n})},
\]
where $f(t,\kappa|\hist_{t_n})$ is the joint density of the time and
the mark (again the word the density is used in a broad sense)
conditional on past times and marks, and $F(t|\hist_{t_n})$ is the
conditional cumulative distribution function of $t$ also conditional
on the past times and marks. Therefore following the same arguments as
in Section~\ref{sec.cif}, the conditional intensity function
$\lambda^*(t,\kappa)$ can now be interpreted for the case of discrete
marks by
\begin{eqnarray*}
\lambda^*(t,\kappa)\dd t = \E[N(\dd t \times \kappa)|\hist_t],
\end{eqnarray*}
that is, the mean number of points in a small time interval $\dd t$
with the mark $\kappa$. Similarly for the continuous case,
\begin{eqnarray*}
\lambda^*(t,\kappa)\dd t\dd \kappa = \E[N(\dd t \times \dd \kappa)|\hist_t],
\end{eqnarray*}
that is, the mean number of points in a small time interval $\dd t$
with the mark in a small interval $\dd\kappa$.

We revisit the Hawkes process from Example~\ref{ex.haw}, now with
marks:

\begin{example}[marked Hawkes process]\label{ex.etas}
  The ETAS (epidemic type aftershock sequence) model is a particular
  type of marked Hawkes process for modelling earthquakes times and
  magnitudes. Here $\kappa_i\in[0,\infty)$ denotes the magnitude of an
  earthquake occurring at time $t_i$. In its simplest form the ETAS
  model can be defined by its ground intensity
  \[
  \lambda^*(t) = \mu +
  \alpha\sum_{t_i<t}e^{\beta\kappa_i}e^{-\gamma(t-t_i)},
  \]
  where $\alpha,\beta,\gamma>0$ are parameters, and an exponential
  distribution as its mark density
  \[
  f^*(\kappa|t) = \delta e^{-\delta \kappa}.
  \]
  Equivalently we could define it by its conditional intensity
  function including both marks and times
  \[
  \lambda^*(t,\kappa) = \left(\mu +
  \alpha\sum_{t_i<t}e^{\beta\kappa_i}e^{-\gamma(t-t_i)}\right) 
  \delta e^{-\delta \kappa}.
  \]
  The idea behind using this model is that earthquakes cause
  aftershocks - this is reflected in the fact that every new
  earthquake increases the intensity by $\alpha
  e^{\beta\kappa_i}$. Note that large earthquakes increase the
  intensity more than small earthquakes. For more on the ETAS model,
  see e.g.\ \cite{ogata-88,ogata-98}.
\end{example}

We sometimes make simplifying independence assumptions on the
marks. An {\em unpredictable mark} is a mark that does not depend on
the past (and therefore cannot be ``predicted'' using the information
about the past, hence the term
``unpredictable''). Example~\ref{ex.etas} has unpredictable marks,
since $f^*(\kappa|t)$ does not depend on the past. An even stronger
assumption is that of an {\em independent mark}, which means that
$\kappa_i$ is independent of everything else except maybe
$t_i$. Example~\ref{ex.etas} does not have independent marks, since
the ground intensity depends on the past marks (which is just another
way of saying that the marks depend on the future events).

\section{Inference}

There are many possibilities for estimating the parameters in a
process specified by a conditional intensity function. The likelihood
function for such a process has a fairly simple expression, which
usually means that maximum likelihood inference or Bayesian inference
are good choices.

\subsection{Likelihood function}

Assume that we have observed a point pattern $(t_1,\ldots,t_n)$ on
$[0,T)$ for some given $T>0$, and if we are in the marked case, also
its accompanying marks $(\kappa_1,\ldots,\kappa_n)$. Furthermore, let the {\em integrated conditional intensity function} (or integrated ground intensity function in the marked case) be given by
\[
\Lambda^*(t) = \int_0^t \lambda^*(s) \dd s.
\]
Then the likelihood function is given by the following proposition.

\begin{proposition}\label{prop.lik}
  Given an unmarked point pattern $(t_1,\ldots,t_n)$ on an observation
  interval $[0,T)$, the likelihood function is given by
  \[
  L = \left( \prod_{i=1}^n \lambda^*(t_i) \right) \exp
  (-\Lambda^*(T)).
  \]
  Given a marked point pattern $((t_1,\kappa_1),\ldots,(t_n,\kappa_n))$
  on $[0,T)\times\mathbb{M}$, the likelihood function is given by
  \[
  L = \left( \prod_{i=1}^n \lambda^*(t_i,\kappa_i) \right) \exp
  (-\Lambda^*(T)).
  \] 
\end{proposition}

\begin{proof}
  The likelihood function is the joint density function of all the points in
  the observed point pattern $(t_1,\ldots,t_n)\in [0,T)$, and can
  therefore be factorised into all the conditional
  densities of each points given all points before it. This yields
  \begin{eqnarray*}
    L = f(t_1|\hist_0) f(t_2|\hist_{t_1}) \cdots f(t_n|\hist_{t_{n-1}}) 
    (1-F(T|\hist_{t_n})),
  \end{eqnarray*}
  where the last term $(1-F(T|\hist_{t_n}))$ appears since the
  unobserved point $t_{n+1}$ must appear after the end of the
  observation interval, and the term $\hist_0$ contains the
  information that there are no events before time 0. Using
  (\ref{eq.int}) and (\ref{eq.fstar}), we get that
  \begin{eqnarray*}
    L &=& \left(\prod_{i=1}^n f(t_i|\hist_{t_{i-1}})\right) 
    \frac{f(T|\hist_{t_n})}{\lambda^*(T)}\\
    &=& \left(\prod_{i=1}^n \lambda^*(t_i) \exp
      \left(-\int_{t_{i-1}}^{t_i} \lambda^*(s) \dd s \right)\right) 
    \exp\left(-\int_{t_n}^T \lambda^*(s) \dd s \right)\\
    &=& \left(\prod_{i=1}^n\lambda^*(t_i)\right) \exp
    \left(-\int_0^T \lambda^*(s) \dd s \right),
  \end{eqnarray*}
  where $t_0=0$. This proves the result for the unmarked case. To
  obtain the result for the marked case, start by the factorisation
  \begin{eqnarray*}
    L &=& f(t_1|\hist_{t_0})f(\kappa_1|t_1,\hist_{t_0}) \cdots 
    f(t_n|\hist_{t_{n-1}})f(\kappa_n|t_n,\hist_{t_{n-1}})
    (1-F(T|\hist_{t_n}))
  \end{eqnarray*}
  All the terms except the conditional mark densities 
  $f(\kappa_i|t_i,\hist_{t_{i-1}})=f^*(\kappa_i|t_i)$ are the same as
  in the unmarked case, so
  \begin{eqnarray*}
    L &=& \left(\prod_{i=1}^n f^*(\kappa_i|t_i)\right)
    \left(\prod_{i=1}^n\lambda^*(t_i)\right) \exp
    \left(-\int_0^T \lambda^*(s) \dd s \right) \\
    &=& \left(\prod_{i=1}^n\lambda^*(t_i,\kappa_i)\right) \exp
    \left(-\int_0^T \lambda^*(s) \dd s \right),
  \end{eqnarray*}
  which establishes the result for the marked case.
\end{proof}

\subsection{Estimation}

Although Proposition~\ref{prop.lik} gives an explicit expression for
the likelihood function, it is rarely simple enough that we can find
the maximum likelihood estimate (MLE) analytically. One special case
where we can find the MLE is the homogeneous Poisson process:

\begin{example}[MLE for the homogeneous Poisson process]
  For the homogeneous Poisson process with intensity $\lambda^*(t) =
  \lambda$ observed on an interval $[0,T)$ for some $T>0$, the likelihood simplifies to
  \[
  L = \lambda^n \exp(-\lambda T).
  \]
  Differentiating this and equating to zero, we get that the MLE is
  given by
  \[
  \hat\lambda = \frac{n}{T}.
  \]
  Note that this expression does not depend on the times of the
  points, only the total number of points. However, this is
  not true for other processes.
\end{example}

For most other point processes we will require numerical methods to
obtain estimates, such as Newton-Raphson for maximizing the
likelihood, or Markov chain Monte Carlo for approximating the
posterior in a Bayesian approach.

\section{Simulation}\label{sec.sim}

Simulation turns out to be fairly easy when the conditional intensity
function is specified. The conditional intensity function leads to
two different approaches for simulating a point process: The inverse
method and Ogata's modified thinning algorithm. Both are
generalisations of similar methods for simulation of inhomogeneous
Poisson processes.

\subsection{Inverse method}\label{sec.inv}

The basic idea in the inverse method is that we simulate a unit-rate
Poisson process (this is just a series of independent exponential
random variables with mean one) and transform these into the desired
point process using the integrated conditional intensity function. The following proposition is the key result behind this method.

\begin{proposition}\label{prop.inverse}
  If $(s_i)_{i\in\mathbb{Z}}$ is a unit rate Poisson process on
  $\mathbb{R}$, and $t_i=\Lambda^{*-1}(s_i)$, then $(t_i)_{i\in\mathbb{Z}}$ is
  a point process with intensity $\lambda^*(t_i)$.
\end{proposition}

\begin{proof}
We prove this by induction, so assume that for $i\leq n$, $s_i$ follows a unit rate Poisson process, and $t_i$ follows a point process with intensity $\lambda^*$. Now consider the next point in both processes, say $S_{n+1}$ and $T_{n+1} = \Lambda^*(S_{n+1})$. Letting $S=S_{n+1}-s_n$ follow a unit rate exponential distribution which is independent of everything else, we need to prove that $T_{n+1}$ follows a point process with intensity $\lambda^*$ or equivalently has the correct distribution function $F(\cdot|\hist_{t_n})$. Denoting the distribution function of $T_{n+1}$ by $F_{T_{n+1}}(t|\hist_{t_n})$, we get that
\begin{eqnarray*}
F_{T_{n+1}}(t|\hist_{t_n})
&=& \p(T_{n+1} \leq t | \hist_{t_n})\\
&=& \p(\Lambda^{*-1}(S+s_n) \leq t | \hist_{t_n})\\
&=& \p(S \leq \Lambda^*(t)-s_n | \hist_{t_n})\\
&=& 1 - \exp(-(\Lambda^*(t)-s_n))\\
&=& 1 - \exp(-(\Lambda^*(t)-\Lambda^*(t_n)))\\
&=& 1-\exp\left(-\int_{t_{n}}^t\lambda^*(u) \dd u\right)\\
&=& F(t|\hist_{t_n}),
\end{eqnarray*}
where we have used that $s_n=\Lambda^*(t_n)$ in the fifth equality, and 
(\ref{eq.Fstar}) in the last one. Thus $T_{n+1}$ follows the correct distribution.
\end{proof}

Although the point process is defined on the whole of $\mathbb{R}$ in
Theorem \ref{prop.inverse}, this condition can be relaxed. If we
instead use a Poisson process with $s_i\in[0,T]$, then we get a new
point process with $t_i\in[0,\Lambda^{*-1}(T)]$, i.e.\ we also need to
transform the final end point. This means we cannot simply simulate a
Poisson process on the interval needed, since this interval changes
during the transformation, so we need to simulate one exponential
variable at a time, and then transform them to see if our simulation
fills out the whole interval. The following algorithm does this.

\begin{algorithm}\label{algo.inv}{\bf (Simulation by inversion)}
  \begin{enumerate}
  \item Set $t=0$, $t_0=0$ and $n=0$ (note that $t_0$ is not an event).
  \item Repeat until $t>T$:
    \begin{enumerate}
    \item Generate $s_n\sim\Exp(1)$.
    \item Calculate $t$, where $t = \Lambda^{*-1}(s_n)$.
    \item If $t < T$, set $n=n+1$ and $t_n=t$.
    \end{enumerate}
  \item Output is $\{t_1,\ldots,t_n\}$.
  \end{enumerate}
\end{algorithm}

The difficult part of this algorithm is of course calculating
$t$ in step 2(b) since this requires finding the inverse of the
integrated conditional intensity function. Notice that since
$\lambda^*$ is non-negative, we get that $\Lambda^*$ is
non-decreasing. Strictly speaking, this means that $\Lambda^*$ may
not even be an invertible function, since it can be constant on
intervals (corresponding to $\lambda^*$ being zero in these
intervals). However, any point $s_i$ from the Poisson process will hit
these points with probability zero, so we never need to evaluate
$\Lambda^{*-1}$, where it is not well-defined.

\begin{example}[Hawkes process, Inverse method]
  We revisit the special case of Hawkes process from
  Example~\ref{ex.haw} given by (\ref{eq.hawexp}). For this we
  get the integrated conditional intensity function
  \[
  \Lambda^*(t) = \mu t + \alpha \sum_{t_i<t}
  \left(1-e^{-(t-t_i)}\right).
  \]
  Looking at the expression, it seems to be hard solve this with
  respect to $t$, so an analytical expression for $\Lambda^{*-1}$
  is not available, meaning we will need to approximate this when we
  use Algorithm~\ref{algo.inv}. A simple way of doing this is to
  calculate $\tilde s_i=\Lambda^*(\tilde t_i)$ starting at very small
  values of $\tilde t_i$ and then increase $\tilde t_i$ until
  $s_i\approx\Lambda^*(\tilde t_i)$, and then use $t_i=\tilde t_i$.
\end{example}

The easiest way to generalise this to the marked case is to simulate
the associated mark to an event $t_i$ just after we have transformed
$s_i$ to $t_i$ (notice that we have all the information that this may
depend on, since we have already simulated the past events and marks).

\subsection{Ogata's modified thinning algorithm}

Ogata's modified thinning algorithm \citep{ogata-81} is a thinning
algorithm based on simulating homogeneous Poisson processes with too
high intensities and then thin out the points that are too many
according to the conditional intensity function. Since the
conditional intensity function depends on the past, we have to do this
starting in the past and follow the direction of time.

The basic idea behind the algorithm is that when we are at time $t$ we
need to find out where to place the next point $t_i>t$. To do this we
simulate a homogeneous Poisson process on some interval $[t,t+l(t)]$ for
some chosen function $l(t)$ (this is the maximum distance we may go
forward in time from $t$ and it may be infinite). This Poisson process
has a chosen constant intensity on $[t,t+l(t)]$, which fulfills
\begin{equation}\label{eq.m}
  m(t)\geq\sup_{s\in[t,t+l(t)]}\lambda^*(s).
\end{equation}
Actually we only need to simulate the first point $t_i$ of this Poisson
process. There are now two possibilities: If $t_i>l(t)$, then there is
no point in $[t,t+l(t)]$, so we start again from $t+l(t)$, but if
$t_i\leq l(t)$, there may be a point at $t_i$ in $[t,t+l(t)]$. In the
latter case we need to figure out whether to keep this point or
not. To get the correct intensity, we keep it with probability
$\lambda^*(t_i)/m(t)$. Whether or not we keep it, we start all over at
$t_i$. 

\begin{algorithm}\label{algo.ogata}({\it Ogata's modified thinning
    algorithm.})
\begin{enumerate}
\item Set t=0 and n=0.
\item Repeat until $t>T$:
  \begin{enumerate}
  \item Compute $m(t)$ and $l(t)$.
  \item Generate independent random variables $s \sim \Exp(m(t))$ and
    $U \sim \text{Unif}([0,1])$.
  \item If $s>l(t)$, set $t = t+l(t)$.
  \item Else if $t+s>T$ or $U>\lambda^*(t+s)/m(t)$, set $t = t+s$.
  \item Otherwise, set $n = n+1$,
    $t_n = t+s$, $t = t+s$.
  \end{enumerate}
\item Output is $\{t_1,\ldots,t_n\}$.
\end{enumerate}
\end{algorithm}

\begin{proposition}\label{prop.ogata}
  The output of Algorithm~\ref{algo.ogata} is a realisation of a
  point process with conditional intensity function $\lambda^*(t)$.
\end{proposition}

\begin{proof}
  It follows from independent thinning that this process has the right
  conditional intensity function (essentially the explanation above
  the algorithm is the proof).
\end{proof}

\begin{example}[Hawkes process, Ogata's modified thinning algorithm]
  In order to use the algorithm we need to choose the $m(t)$ and
  $l(t)$, and the only requirement is that the inequality (\ref{eq.m})
  is fulfilled at any possible step of the algorithm. Since
  \[
  \lambda^*(t) = \mu + \alpha\sum_{t_i<t}\exp(-(t-t_i)),
  \]
  is non-increasing (except when new points appear), we can choose
  $m(t)=\lambda(s)$ at every starting point $s$ in the algorithm and any $t\geq s$, and
  $l(t)=\infty$. This choice can be used for any point process where
  $\lambda^*(t)$ only increases when new points arrive. So the Hawkes
  process can be simulated either by the inverse method or Ogata's
  modified thinning algorithm (but in fact there are simpler methods
  for simulating the Hawkes process, see e.g.\
  \cite{moller-rasmussen-05,moller-rasmussen-06}).
\end{example}

It is easy to generalise the algorithm to the marked case: every time
we keep a point $t_i$ in the algorithm, we should simulate its marks
from the mark distribution $f^*(\kappa_i|t_i)$ (just as for the
inverse method we have the required knowledge of the past when we need
to simulate this).

\subsection{Why simulate a point process?}\label{sec.whysim}

Simulations of point processes are useful for many things:

{\em What does a point pattern typically look like?} Simulating a point
process a couple of times for a given model and a given set of
parameters will provide valuable information on what a typical point
pattern looks. Is it clustered or regular? Is it inhomogeneous or
homogeneous? Does it look anything remotely like the data you are
going to spend the next week fitting the model to?

{\em Prediction:} Given an observed past, what does the future hold? The
specification of the conditional intensity function means that it is
easy to include the already observed past, and then simulate the
future.

{\em Model checking:} Prediction can also be used for model checking if we
only use the data in the first half of the observation interval to fit a model,
and then simulate predictions of the second half to see if this
corresponds to the second half of the observed data. Or we can use all
of the data, and compare with simulations of the whole dataset.

{\em Summary statistics:} Many quantities can be calculated explicitly from
the conditional intensity function, such as the probability of getting
no events in the next month or the mean time to the next
event. However, particularly complicated summary statistics may not be
available on closed form, but can instead be approximated by
simulation. For example, the mean number of events in a given time
interval may not be available on closed form for a complicated model,
but we can then approximate it by the average number of points in a
number of simulations.

\section{Model checking}

In addition to the model checking approaches mentioned in
Section~\ref{sec.whysim}, there is a particular kind of model checking
associated with the conditional intensity function known as residual
analysis. 

\subsection{Residual analysis}

Residual analysis \citep{ogata-88} is a type of model checking for
point processes specified by a conditional intensity function. It is
based on the reverse transformation than the one used in
Proposition~\ref{prop.inverse}.

\begin{proposition}\label{prop.resid}
  If $(t_i)_{i\in\mathbb{Z}}$ is a point process with intensity
  $\lambda^*(t_i)$, and $s_i=\Lambda^{*}(t_i)$, then
  $(s_i)_{i\in\mathbb{Z}}$ is a unit rate Poisson process.
\end{proposition}

\begin{proof}
	This is proved in a similar manner as Proposition~\ref{prop.inverse}.
\end{proof}

Thus if a point pattern is a realization of a point process with
conditional intensity function $\lambda^*$, then the integrated
conditional intensity function will transform the pattern into a
realization of a unit rate Poisson process. In practice this means
that if we have modelled an observed point pattern with a point
process, and the type of point process is well-chosen, then the
transformed pattern should closely resemble a unit-rate Poisson
process. In other words, the model checking boils down to checking
whether the interevent times are independent exponential variables
with mean one.

If the model does not fit, residual analysis may provide important
information on how it does not fit. For example, if the data contains
an unrealistically large gap for the model between $t_i$ and
$t_{i+1}$, then the transformed data will contain a large gap between
$s_i$ and $s_{i+1}$, i.e.\ $s_{i+1}-s_i$ will be to large to
realistically come from a unit rate exponential
distribution. A bit of creativity in analysing the residuals can give us
all kinds of information about the original point pattern.

\section{Concluding remarks}

We have now seen that the conditional intensity function is a valuable
tool for point process modelling, and can be used at all stages of
data analysis:
\begin{itemize}
\item Preliminary analysis (simulation of potential models)
\item Model specification and interpretation.
\item Parameter estimation (maximum likelihood or Bayesian estimation).
\item Model checking (residual analysis or simulation based
  approaches).
\item Prediction.
\end{itemize}
However, we should note that basing parameter estimation and model
checking on the same functions of the data is usually considered bad
practice. For example, if we fit a model using maximum likelihood
estimation, we have essentially fitted the conditional intensity
function as well as we can, and it should not come as a surprise if
the residuals fit rather well, since they are also based on the
conditional intensity function. Here it would be more appropriate to
base the model checking on other aspects of the model (such as the
summary statistics given for example in
\cite{moller-waagepetersen-04}), which may not be caught so well by the
conditional intensity function.

\bibliographystyle{natbib}
\bibliography{bibliography}

\end{document}